\documentclass[11pt,onecolumn,draftclsnofoot]{IEEEtran} %!PN

\usepackage{graphicx, cite, verbatim, url}

\newtheorem{subsec:coding}{subsec:coding}

\newtheorem{propo}{Proposition}

\newcommand{\ls}[1]  %% 1 in brackets means \ls takes 1 argument
   {\dimen0=\fontdimen6\the=#1\dimen0
    \advance\lineskip.5\fontdimen5\the\lineskip-\dimen0
    \lineskiplimit=.9\lineskip
    \baselineskip=\lineskip
    \advance\baselineskip\dimen0
    \normallineskip\lineskip
    \normallineskiplimit\lineskiplimit
    \normalbaselineskip\baselineskip
    \ignorespaces
   }
   
\begin{document}

%\title{A Medium Access Control Protocol for Cognitive Radio Networks with Guaranteed Throughput Performance}
%\title{A Channel Sensing Error-Aware Medium Access Control Protocol for Cognitive Radio Networks}
%\title{Design and Analysis of A Sensing Error-Aware Medium Access Control Protocol for Cognitive Radio Networks}
\title{A Sensing Error Aware MAC Protocol for Cognitive Radio Networks}

%\title{Scheduling in Cognitive Radio Networks for Guaranteed Throughput Performance}

\author{Donglin~Hu~and~Shiwen~Mao
\thanks{D. Hu and S. Mao are with the Department of Electrical and Computer Engineering, Auburn University, Auburn, AL 36849-5201. Email: dzh0003@tigermail.auburn.edu, smao@ieee.org.}
\thanks{Part of this work was presented at IEEE GLOBECOM 2009~\cite{Hu09gc}.} 
}

\maketitle

\begin{abstract}
Cognitive radios (CR) are intelligent radio devices that can sense the radio environment and adapt to changes in the radio environment. Spectrum sensing and spectrum access are the two key CR functions. In this paper, we present a spectrum sensing error aware MAC protocol for a CR network collocated with multiple primary networks.  We explicitly consider both types of sensing errors in the CR MAC design, since such errors are inevitable for practical spectrum sensors and more important, such errors could have significant impact on the performance of the CR MAC protocol.  Two spectrum sensing polices are presented, with which secondary users collaboratively sense the licensed channels.  The sensing policies are then incorporated into $p$-Persistent CSMA to coordinate opportunistic spectrum access for CR network users.  We present an analysis of the interference and throughput performance of the proposed CR MAC, and find the analysis highly accurate in our simulation studies.  The proposed sensing error aware CR MAC protocol outperforms two existing approaches with considerable margins in our simulations, which justify the importance of considering spectrum sensing errors in CR MAC design.
\end{abstract}


\thispagestyle{plain}\pagestyle{plain}
%\thispagestyle{empty}\pagestyle{empty}

%\ls{1.8}

%-----------------------------------------------
\section{Introduction}
%-----------------------------------------------

A {\em cognitive radio} (CR) is a frequency-agile wireless communication device with a monitoring interface and intelligent decision-making that enables dynamic spectrum access~\cite{Zhao07c}. A CR can sense the radio environment and adapt to changes in the radio environment. 
%Spectrum sensing and spectrum access are the two key CR functions. 
The CR concept represents a significant paradigm change in spectrum regulation and utilization, i.e., from exclusive use of spectrum by licensed users (or, {\em primary users}) to dynamic spectrum access for unlicensed users (or, {\em secondary users}).  The high potential of CRs has attracted considerable efforts from the wireless community recently, 
%in the development of better spectrum management policies and techniques~\cite{Zhao07c, Zhao09}.
for developing more efficient spectrum management policies and techniques~\cite{Zhao07c, Zhao09}.

%Recent studies have shown that the spectrum allocated by traditional fixed policy is getting scarce. Measurements suggest that spectrum utilization in these fixed allocated frequency bands is very low and inefficiency. Of course, adoption of efficient modulation and coding techniques can improve spectrum utilization, but cannot solve the inefficiency problem~\cite{Zhao07c, Zhao09}.

%Although the basic concept of cognitive radio is simple, the efficient design of cognitive network is challenging. Specifically, identifying the available transmission opportunities results in amount of design problems to the medium access control (MAC) layer. Two most important design problems of all is (i) how the secondary users decide when and which channel they should tune to and (ii) how they identify the availability of channel when there exist sensing errors.

Quality of service guarantee in wireless networks is a challenging problem that has attracted tremendous efforts~\cite{Tang07JSAC, Tang08TW, Zhang06CM, Tang07TW}.  Although the basic concept of CR is intuitive, it is challenging to design efficient cognitive network protocols to fully capitalize CR's potential.  In order to exploit transmission opportunities in licensed bands, the tension between primary user protection and secondary user spectrum access should be judiciously balanced.  Spectrum sensing and spectrum access are the two key CR functions.  Important design factors include (i) how to identify transmission opportunities, (ii) how secondary users determine, among the licensed channels, which channel(s) and when to access for data transmission, and (iii) how to avoid harmful interference to primary users under the omnipresent of spectrum (or, channel) sensing errors.  These are the problems that should be addressed in the medium access control (MAC) protocol design for CR networks.  
%The challenge stems from the dynamic channel processes and imperfect channel sensing. 
Although very good understandings on the availability process of licensed channels have been gained recently~\cite{Motamedi07, Geirhofer08JNL}, there is still a critical need to develop analytical models that take channel sensing errors into account for guiding the design of CR MAC protocols. 

In this paper, we present a channel sensing error aware MAC protocol for a CR network collocated with multiple primary networks.  We assume primary users access the licensed channels following a synchronous time slot structure~\cite{Zhao07c, Su08JNL}.  The channel states are independent to each other and each evolves over time following a discrete-time Markov process~\cite{Motamedi07, Zhao07c}.  Secondary users use their software-defined radio (SDR)-based transceivers to tune to any of the licensed channels, to sense and estimate channel status and to 
%transmit and receive data from 
access the channels when they are found (or, believed) to be available. 
We explicitly consider channel sensing errors in the design of the CR MAC protocol. It has been shown in prior work that generally there are two types of channel sensing errors: (i) {\em false alarm}, when an idle channel is identified as busy, thus a spectrum opportunity will be wasted, (ii) {\em miss detection}, when a busy channel is identified as idle, thus leading to collision with primary users, since CR users will attempt to use such ``idle'' channels.  We consider both types of spectrum sensing errors in our CR MAC design, which have been shown to be unavoidable for practical spectrum sensors~\cite{Zhao07c}.

In particular, we develop two channel sensing polices, with which secondary users collaboratively sense the licensed channels and predict channel states. With the {\em memoryless sensing} policy, each secondary user chooses one of the $M$ licensed channels to sense 
%(and to access after a successful sensing phase) 
with equal probability.  During the sensing phase, secondary users also exchange sensing results through a separate control channel.  This sensing policy is further improved with a mechanism to spread out secondary users to sense different channels, therefore reducing the chance that a channel is not sensed by any of the users. When spreading out secondary users to the channels, the mechanism also 
% and by considering 
considers the autocorrelation of channel processes to obtain more accurate sensing results.  This is termed {\em improved sensing} policy.  
%We present a closed-form analysis of both sensing policies. 

These two sensing polices are then incorporated into the $p$-Persistent Carrier Sense Multiple Access (CSMA) mechanism to make sensing error aware CR MAC protocols.  We analyze the proposed CR MAC protocols with respect to the interference and throughput performance and derive closed-form expressions.  Primary user protection is achieved via tunning the channel access probability $p$ of $p$-Persistent CSMA according to the interference analysis. The CR MACs also aims to maximize the CR network throughput while satisfying the primary user protection constraints. Through simulations, we find that the analysis is highly accurate as compared to simulation results. In addition, the proposed sensing error aware CR MAC protocols outperform two existing schemes with considerable gain margins, which justify the importance of considering channel sensing errors in CR MAC design. 

%A promising approach, known as Opportunistic Spectrum Access (OSA), allows users in the secondary network to access channels when users in the primary networks are not using the channels~\cite{Zhao07c}. These two kinds of networks share the same spectrum. The users in these two networks are regarded as primary users and secondary users respectively. Primary users are able to access the spectrum at any time as they wish, while the secondary users have to wait for unused spectrum block to transmitting data. These frequency opportunities are called the ``spectrum hole'' (or ``tile'') .

%CR applications to capitalize the improved spectrum efficiency, with maximizing secondary user performance with respect to throughput or application layer metrics~\cite{Hu09}, while keep the harmful interference to primary users below a prescribed threshold. 

The remainder of this paper is organized as follows.  We describe the network model and assumptions in Section~\ref{sec:model}.  We then present the proposed CR MAC protocol and analyze its performance in Section~\ref{sec:crmac}. Our simulation studies are presented in Section~\ref{sec:sim}. Section~\ref{sec:related} discusses related work and Section~\ref{sec:con} concludes this paper.

%-----------------------------------------------
\section{Network Model and Assumptions \label{sec:model}}
%-----------------------------------------------

\begin{figure}
\centering
\includegraphics[width=3.4in]{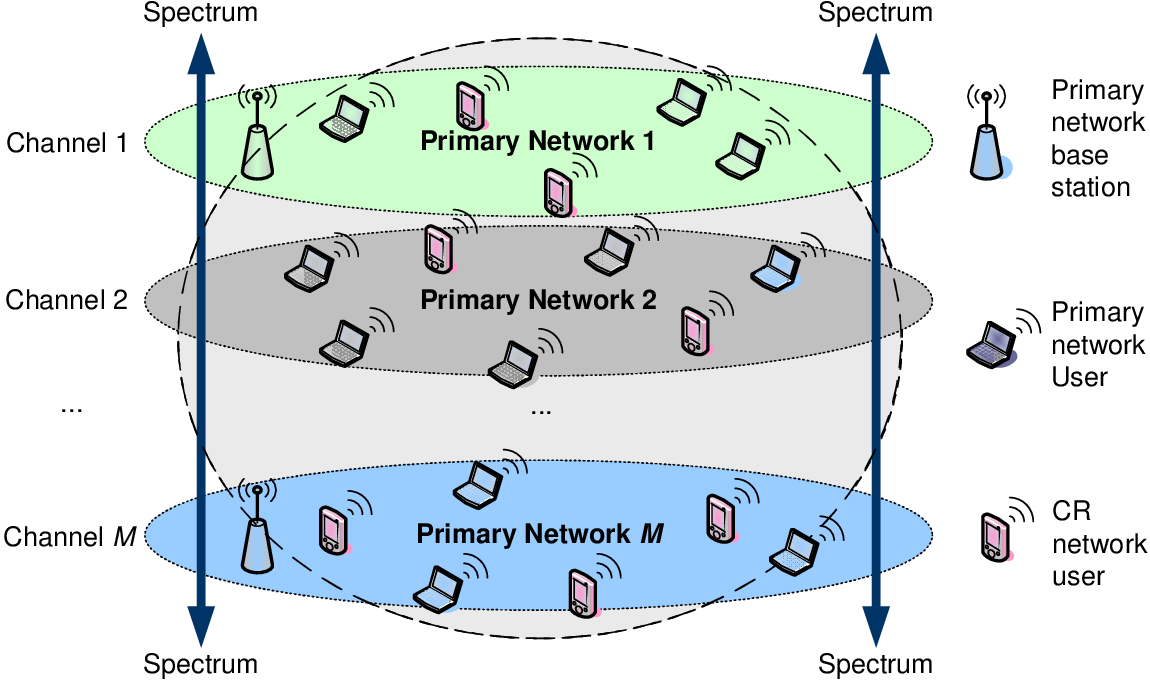}
\caption{The CR secondary network is collocated with $M$ primary networks.}
\label{fig:networkarc}
\end{figure}

The network model considered in this paper is illustrated in Fig.~\ref{fig:networkarc}.  
Consider $M$ primary networks, each allocated with a licensed channel.  
%We assume that each of the $M$ channels has the same bandwidth as others~\cite{Su08JNL}.  
%As shown in Fig.~\ref{fig:ChanModel}, we 
We assume the primary users access the channels following a synchronous slot structure as in prior work~\cite{Zhao07c, Su08JNL, Hu09}.  The channel states are independent to each other and each of the $M$ channels evolves over time following a discrete-time two-state Markov process, as shown in Fig.~\ref{fig:ChanModel}. Such channel model has been validated by recent measurement studies~\cite{Zhao07c, Motamedi07, Su08JNL}.
%There is a spectrum band consisting of $M$ channels, each evolving over time independently. 
%For ease of explanation, we assume that the spectrum band is continuous and equally divided into the $M$ channels.  Each of the $M$ channels is allocated to a different primary network, while the channel states are independent to each other.  As shown in Fig.~\ref{fig:ChanModel}, we assume the primary users access the channels following a synchronous slot structure; the state of each of the $M$ channels evolves over time following a discrete-time two-state Markov process~\cite{Zhao07c}. 
We define the {\em network state vector} in slot $t$ as $\vec{S}(t)=[S_1(t),S_2(t),\dots,S_M(t)]$, where $S_m(t)$ denotes the state of channel $m$, for $m=1, 2, \cdots, M$. When channel $m$ is idle, we have $S_m(t)=0$; when channel $m$ is busy, we have $S_m(t)=1$. 

\begin{figure}
\centering
\includegraphics[width=2.8in]{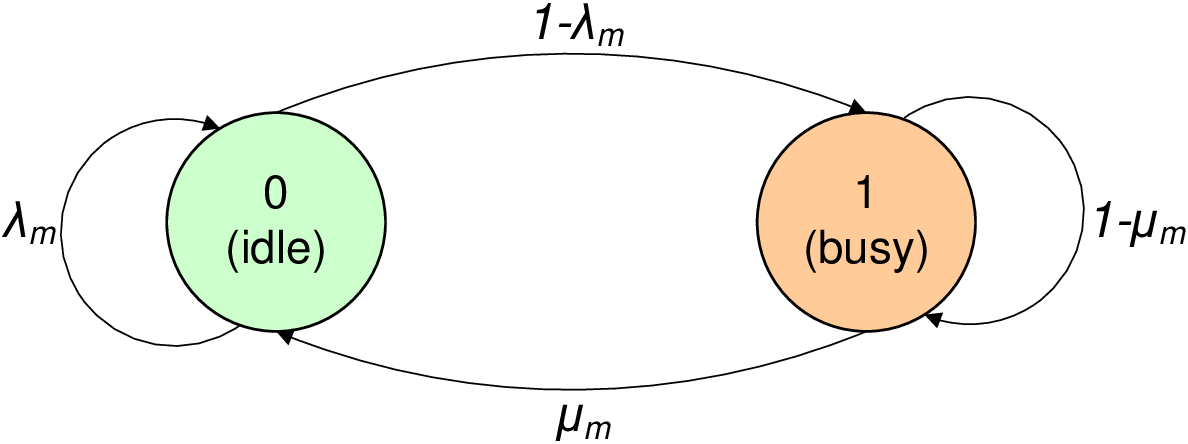}
\caption{The discrete-time two-state Markov model for the state of channel $m$, $S_m$, for $m = 1,2,\dots,M$.}
\label{fig:ChanModel}
\end{figure}

Let $\lambda_m$ and $\mu_m$ be the transition probability of remaining in state $0$ and the transition probability from state $1$ to $0$ for channel $m$, respectively. Let $\eta_m = \Pr(S_m=1)$ denote the {\em utilization} of channel $m$ with respect to primary user transmissions.
%, denoted as , can be derived as
%We have
%\begin{eqnarray}\label{eq:ProbBusy}
%  \eta_m = \lim_{T\to \infty} \frac{1}{T} \mbox{$\sum_{t=1}^T$} S_m(t) 
%         = \frac{1-\lambda_m}{1-\lambda_m+\mu_m}.
%\end{eqnarray}
Let $\zeta_m = \Pr(S_m=0)$ be the probability that channel $m$ is idle (i.e., not being used by primary users). We then have
\begin{eqnarray}
  \eta_m & = & \lim_{T\to \infty} \frac{1}{T} \sum_{t=1}^T S_m(t) 
  %\\ & = & 
  = \frac{1-\lambda_m}{1-\lambda_m+\mu_m} \label{eq:ProbBusy} \\
  \zeta_m & = & 1 - \Pr(S_m=1) = \frac{\mu_m}{1 - \lambda_m + \mu_m}. \label{eq:ProbIdle}
\end{eqnarray}

We assume a secondary network collocated with the $M$ primary networks, within which $N$ secondary users take advantage of the spectrum white spaces in $M$ licensed channels for data transmissions.  For protection of primary users, 
% on channel $m$, 
the probability of collision caused by secondary user transmissions to primary users should be upper bounded by a prescribed threshold $\gamma_m$, for $m=1,2,\cdots,M$.  

As in prior work~\cite{Su08JNL, Nan07, Motamedi07}, we assume that each secondary user is equipped with two transceivers: a {\em control transceiver} that operates over a dedicated control channel, which we assume is always available (e.g., a channel in the industrial, scientific and medical (ISM) band), and a {\em data transceiver} that is used for data communications through the $M$ licensed channels.  The data transceiver consists of an SDR that can be tuned to any of the $M$ licensed channels to transmit and receive data.  Secondary users also use their transceivers for spectrum sensing and exchanging sensing results. 
%to obtain information of available channels.  
% from there.

%------------------------------------------------------
\section{Sensing Error Aware CR MAC Protocol \label{sec:crmac}}
%------------------------------------------------------

For the CR network described in Section~\ref{sec:model}, we develop sensing aware MAC protocols for opportunistic spectrum access. 
% that considers sensing errors.  
The time slot structure of the proposed MAC protocols is shown in Fig.~\ref{fig:SlotStructure}, which consists of a {\em sensing phase} and a {\em transmission phase}.  The sensing phase is further divided into $\bar{K}$ mini-slots, within which each secondary user senses one of the licensed channels.  CR users access the channels for data transmission during the transmission phase.  Let $T_s$, $T_{ms}$, and $T_{data}$ denote the duration of a time slot, a mini-slot, and the transmission phase, respectively (see Fig.~\ref{fig:SlotStructure}), we have
\begin{eqnarray} \label{eq:ts}
   T_s = \bar{K} \times T_{ms} + T_{data}. 
\end{eqnarray}

We first discuss the two key components of the proposed protocols, i.e., channel sensing and channel access, and then analyze their performance with respect to primary user protection and the expected throughput.  Table~\ref{tb:Notation} summarizes the notation used in this paper.

\begin{figure}
\centering
\includegraphics[width=3.3in]{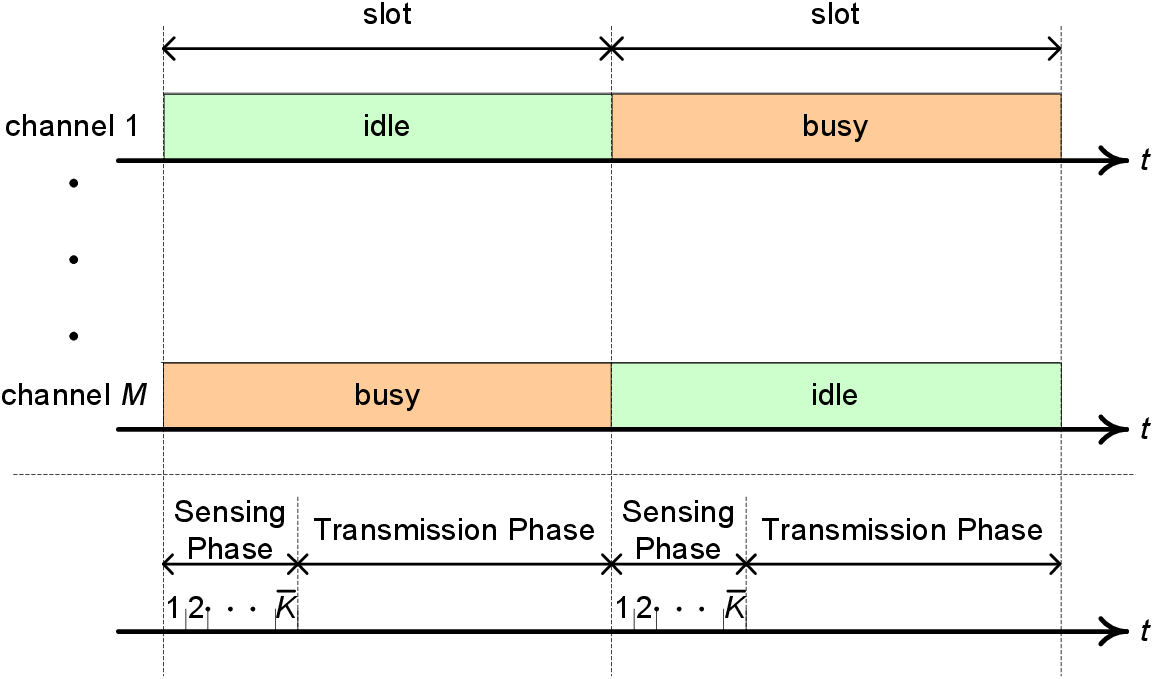}
\caption{The time slot structure of the proposed sensing error aware CR MAC protocol.}
\label{fig:SlotStructure}
\end{figure}

%---------------------------------------------------%
\begin{table}
\begin{center}
\caption{Notation}
\begin{tabular}{l l}
\hline
{\em Symbol} & {\em Definition} \\
\hline
$M$ & number of data channels \\
$N$ & number of secondary users \\
$\lambda_m$ & transition probability of channel $m$ from idle to idle  \\
$\mu_m$ & transition probability of channel $m$ from busy to idle  \\
$\eta_m$ & probability that channel $m$ is busy \\
$\zeta_m$ & probability that channel $m$ is idle \\
$\gamma_m$ & maximum allowable collision probability of channel $m$\\
$T_s$ & duration of a time slot \\
$T_{ms}$ & duration of a mini-slot \\
$T_{data}$ & duration of the transmission phase \\
$a_{m,k}$ & probability  that channel $m$ is idle conditioned on $k$ sensing results  \\
$\Theta_0, \Theta_1$ & thresholds for channel decision $\Theta_0 < \Theta_1$ \\
$\Psi^0_{m,k}$ & set of $0$ observations that make $a_{m,k}$ below
$\Theta_0$ \\
$\Psi^1_{m,k}$ & set of $0$ observations that make $a_{m,k}$ above
$\Theta_1$ \\
$\Psi^2_{m,k}$ & set of $0$ observations that make $a_{m,k}$ between  %\\
%              & \ \ 
         $\Theta_0$ and $\Theta_1$ \\
$S_m(t)$ & state of channel $m$ at time $t$\\
$W_{m,i}$ & the $i$th sensing result on channel $m$ \\
$\theta_{m,i}$ & the $i$th observed sensing result on channel $m$ (0 or 1)  \\
$\epsilon_m$ & probability of false alarm on channel $m$ \\
$\delta_m$ & probability of miss detection on channel $m$ \\
$K_m$ & stopping time in the sensing phase for channel $m$ \\
$p$ & transmission probability of a secondary user \\
$P_m^{idle}$ & probability that no secondary user transmits on channel $m$ \\
$P_m^{succ}$ & probability that one secondary user wins 
%\\
%            & \ \ 
   channel $m$ \\
$P_m^{coll}$ & probability of collision on 
%that more than one secondary user transmit on 
%\\
%            & \ \ 
   channel $m$ \\
$P^{idle}$ & probability that no secondary user transmits on control channel \\
$P^{succ}$ & probability that one secondary user wins all the idle channels \\
$P^{coll}$ & probability of collision on control channel \\   
$P_m^{intf}$ & probability of interference from secondary users 
%to primary 
%\\
%            & \ \ 
%   users 
on channel $m$ \\
%$\Delta_m(u)$ & throughput of channel $m$ that is being sensed by $u$ users \\
$\Lambda^1_m(u)$ & throughput of channel $m$ sensed by $u$ users in Case 1 \\
$\Lambda^2_m(u)$ & throughput of channel $m$ sensed by $u$ users in Case 2 \\
$\Omega_1$ & CR network throughput in Case 1 \\
$\Omega_2$ & CR network throughput in Case 2 \\
$R_m$ & data rate of licensed channel $m$ \\
\hline
\end{tabular}
\label{tb:Notation}
\end{center}
\end{table}
%------------------------------------------------------%

%------------------------------------------------------
\subsection{Sensing Phase \label{subsec:sensingph}}

The first key element of the proposed MAC protocols is spectrum, or channel sensing. Although precise and timely channel state information is highly desirable for opportunistic spectrum access and primary user protection, contiguous full-spectrum sensing is both energy inefficient and hardware demanding.  Since we assume a secondary user is equipped with one transceiver for spectrum sensing, i.e., the data transceiver with SDR capability, only one of the licensed channels can be sensed by the secondary user at a time.  

During the sensing phase (see Fig.~\ref{fig:SlotStructure}), a secondary user picks a licensed channel and keeps on sensing it for one or multiple mini-slots.  As discussed, two kinds of detection errors may occur: false alarm and miss detection. 
%With a false alarm, a spectrum opportunity will be wasted, while a miss detection may lead to collision with primary users. 
We assume all secondary users have the same probability of detection errors when sensing channel $m$, $m=1, 2, \cdots, M$. Let $\epsilon_m$ and $\delta_m$ denote the probabilities of false alarm and miss detection on channel $m$, respectively.  The spectrum sensing performance can be represented by the Receiver Operation Characteristic (ROC) curve, where $(1-\delta_m)$ is plotted as a function of $\epsilon_m$~\cite{Zhao07c}.  
For a specific channel $m$ in a certain time slot $t$, the sensing error probabilities can be written as:
\begin{eqnarray}\label{eq:ErrorProb}
	&& \Pr(W_{m,i}=1 \;|\; S_m=0)=\epsilon_m, \mbox{for all } \; i=1,2,\cdots \\
	&& \Pr(W_{m,i}=0 \;|\; S_m=1)=\delta_m, \mbox{for all } \; i=1,2,\cdots,
\end{eqnarray}
where $W_{m,i}$ is the $i$th sensing result of channel $m$ and $S_m$ is state of channel $m$. 

We assume that the sensing results from different users are independent and the sensing results in different mini-slots are also independent to each other. Suppose a secondary user continues to sense channel $m$ for $k$ mini-slots and obtains $k$ sensing results. The conditional probability that channel $m$ is available after the $k$th sensing mini-slot, denoted by $a_{m,k}$, can be derived as
\begin{eqnarray}\label{eq:AvailProb}
a_{m,k} & = & \Pr(S_m=0 \;|\; W_{m,1}=\theta_{m,1},\cdots,W_{m,k}=\theta_{m,k})\nonumber \\
& = & \frac{\Pr(W_{m,i}=\theta_{m,i},i=1,\cdots,k|S_m=0)\Pr(S_m=0)}{\sum_{j=0}^1 \Pr(W_{m,i}=\theta_{m,i},i=1,\cdots,k|S_m=j)\Pr(S_m=j)} \nonumber\\
& = & \frac {\Pr(S_m=0)\prod_{i=1}^k\Pr(W_{m,i}=\theta_{m,i}|S_m=0)}{\sum_{j=0}^1 \Pr(S_m=j)\prod_{i=1}^k\Pr(W_{m,i}=\theta_{m,i}|S_m=j)} \nonumber \\
& = & \left[ 1 + \frac{\Pr(S_m=1)}{\Pr(S_m=0)} \prod_{i=1}^k \frac{\Pr(W_{m,i}=\theta_{m,i}|S_m=1)}{\Pr(W_{m,i}=\theta_{m,i}|S_m=0)} \right]^{-1} \nonumber \\
& = & \left[ 1 + \alpha_m^{d_m} \beta_m^{k-d_m} \frac{\Pr(S_m = 1)}{\Pr(S_m = 0)} \right]^{-1} \nonumber \\
%&& \hspace*{-0.3in} 
& = & \left( 1 + \alpha_m^{d_m} \beta_m^{k-d_m} \frac{\eta_m}{\zeta_m} \right)^{-1}_{,}
\end{eqnarray}
where $d_m$ is the number of observations whose sensing result is $0$ on channel $m$, 
%that the sensing result on channel $m$ is $0$, 
and $\alpha_m$ and $\beta_m$ are defined as follows.
\begin{eqnarray}
\alpha_m = \frac{\Pr(W_{m,i}=0 | S_m=1)}{\Pr(W_{m,i}=0 |
      S_m=0)}=\frac{\delta_m}{1-\epsilon_m}, \; \mbox{for } \theta_{m,i}=0 \\
\beta_m = \frac{\Pr(W_{m,i}=1 | S_m=1)}{\Pr(W_{m,i}=1 |
      S_m=0)}=\frac{1-\delta_m}{\epsilon_m}, \; \mbox{for } \theta_{m,i}=1. 
      \hspace{-0.05in}
\end{eqnarray}

For the secondary user, it is also possible that it obtains some of the $k$ sensing results by local measurements, and receives the remaining sensing results from the control channel in the case that some other secondary users are sensing the same channel $m$. By abuse of notation, we also use $a_{m,k}$ to denote the conditional channel availability probability in this case, due to independence of the sensing results. 
%By abuse of notation, we also use $a_{m,k}$ to denote the conditional probability that channel $m$ is available when there are $k$ sensing results available. That is, some sensing results may be the outcomes of the local sensor, while others may be received from other secondary users through the control channel. Therefore, (\ref{eq:AvailProb}) is also for the case of one user that senses channel $m$ for $k$ mini-slots (i.e., there are $k*1$ sensing samples available).
We plot $a_{m,k}$ as a function of $k$ for the channel idle and busy cases in Fig.~\ref{fig:amk}, using the same parameters as one of the simulations (see Section~\ref{sec:sim}). We have the following proposition for $a_{m,k}$. 

\begin{propo}
When channel $m$ is idle, $a_{m,k}$ is a monotone increasing function of $k$; when channel $m$ is busy, $a_{m,k}$ is a monotone decreasing function of $k$.
\end{propo}
\begin{proof}
From the defintion of $a_{m,k}$ in (\ref{eq:AvailProb}), it follows that
\begin{eqnarray}
&&\Pr(a_{m,k}\ge\theta_1) \nonumber \\
&=&\Pr \left( \left( 1 + \frac{\eta_m}{\zeta_m} \left(\frac{\delta_m}{1-\epsilon_m} \right)^{\sum_{i=1}^k\bar{W}_{m,i}} \left(\frac{1-\delta_m}{\epsilon_m} \right)^{\sum_{i=1}^k W_{m,i}} \right)^{-1}\ge\theta_1 \right) \nonumber \\
&=&\Pr \left( \left(\frac{\delta_m}{1-\epsilon_m} \right)^{\sum_{i=1}^k\bar{W}_{m,i}} \left(\frac{1-\delta_m}{\epsilon_m} \right)^{\sum_{i=1}^k W_{m,i}}\le \left(\frac{1}{\theta_1}-1 \right) \frac{\zeta_m}{\eta_m} \right) \nonumber \\
&=&\Pr \left( \sum_{i=1}^k \left( W_{m,i}\log \left( \frac{1-\delta_m}{\epsilon_m} \right) - \bar{W}_{m,i} \log \left( \frac{1-\epsilon_m}{\delta_m} \right) \right) \le \chi_m \right)
\end{eqnarray}
where $\bar{W}_{m,i}=1-W_{m,i}$ and $\chi_m=\log((\frac{1}{\theta_1}-1)\frac{\zeta_m}{\eta_m})$. 

Since $\epsilon_m <0.5$ and $\delta_m<0.5$ for practical sensors, both $\log \left( \frac{1-\delta_m}{\epsilon_m} \right)$ and $\log \left( \frac{1-\epsilon_m}{\delta_m} \right)$ are positive. If $S_m(t)=0$, we have that $\Pr(W_{i,k+1}=1)<\Pr(W_{i,k+1}=0)=\Pr(\bar{W}_{i,k+1}=1)$. 
%Therefore, 
It follows that $\Pr(a_{m,k}\ge\theta_1) < \Pr(a_{m,k+1}\ge\theta_1)$. That is, $a_{m,k}$ is a monotone increasing function of $k$. 

Similarly, we can show that $\Pr(a_{m,k}\le\theta_0) < \Pr(a_{m,k+1}\le\theta_0)$ when $S_m(t)=1$. That is, $a_{m,k}$ is a monotone decreasing function of $k$ when the channel is busy. 
\end{proof}

\begin{figure}
\centering
\includegraphics[width=3.3in]{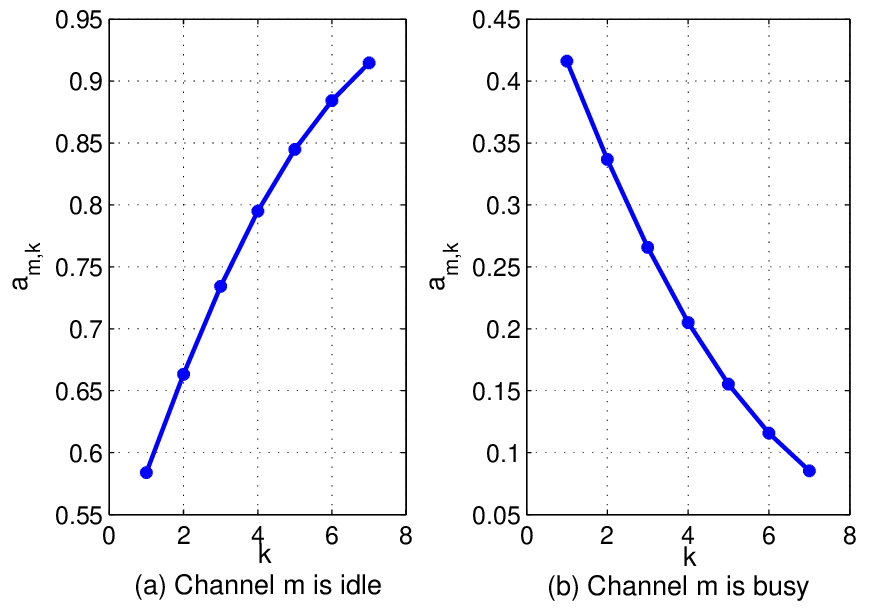}
\caption{Illustration of $a_{m,k}$ as a monotone function of $k$, when $\epsilon_m=0.3$, $\delta_m=0.3$, and $\bar{K}=7$.} 
\label{fig:amk}
\end{figure}

During the sensing phase, each secondary user chooses one channel to sense with equal probability at the beginning of the time slot. Secondary users also report their sensing results over the control channel, and share the corresponding channel sensing results during the mini-slots. Two threshold probabilities $\Theta_0 < \Theta_1$ are used for decision making.  
\begin{itemize}
	\item If the availability of channel $m$, i.e., $a_{m,k}$, is below $\Theta_0$, the channel is believed to be busy and the secondary users will wait till the next time slot to start sensing again.  
	\item If the availability of channel $m$ is between $\Theta_0$ and $\Theta_1$, 
%does not go above the threshold $\Theta_1$ or below the threshold $\Theta_0$, 
secondary users will keep on sensing the same channel to obtain more sensing results for more accurate estimation of the channel state, until the maximum number of mini-slots, $\bar{K}$, is reached.  
	\item If the availability of channel $m$ exceeds $\Theta_1$, the channel is believed to be idle and the secondary users stop sensing and prepare to access the channel (see Section~\ref{subsec:transph}).
\end{itemize}
The {\em stop time} $K_m$ when secondary users stop sensing channel $m$, is a random variable that takes value between $1$ and $\bar{K}$, the maximum number of mini-slots that can be used for sensing (see Fig.~\ref{fig:SlotStructure}).  If we have $\Theta_0 < a_{m,k} < \Theta_1$ by the end of the sensing phase, then channel $m$ state is not identified due to lack of time (or sensing results) and the channel will not be accessed. 

When there are $k$ sensing results available (e.g., one user senses channel $m$ for $k$ mini-slots, or it senses channel $m$ for less than $k$ mini-slots and receives some channel $m$ sensing results from other secondary users), we define three sets of estimates for the state of channel $m$, as: 
\begin{eqnarray}
\Psi^0_{m,k}&=&\{d_m\;|\;a_{m,k}\le \Theta_0, \forall \; 0 \le d_m \le k\} \label{eq:Psi0} \\
\Psi^1_{m,k}&=&\{d_m\;|\;a_{m,k}\ge \Theta_1, \forall \; 0 \le d_m \le k\} \label{eq:Psi1} \\
\Psi^2_{m,k}&=& \{d_m\;|\;\Theta_0 < a_{m,k} < \Theta_1, \forall \; 0 \le d_m \le k\}
           %\nonumber \\
           % &=& 
            = \overline{(\Psi^0_{m,k} \cup \Psi^1_{m,k})}, \label{eq:Psi2}
\end{eqnarray}
where $d_m$ is the number of observations 
%that the sensing result on channel $m$ is $0$. 
whose sensing result is $0$ on channel $m$. We then present two channel sensing policies based on this classification in the following.

%------------------------------------------------------
\subsubsection{Memoryless Sensing Policy \label{subsubsec:memoryless}}

We first propose a memoryless sensing policy in this section, with which secondary users cooperatively sense the licensed channels.  We call the policy ``memoryless'' since it does not consider the channel sensing and access results in the previous time slot for simplicity. 
%Under this policy, secondary users cooperatively sense the licensed channels.  
With this memoryless policy, each secondary user chooses one of the $M$ licensed channels to sense with equal probability, i.e., $1/M$. Furthermore, channel selections of the $N$ secondary users are independent and identically distributed (i.i.d.).  

Let $U_m$ be the random variable representing the number of secondary users that select channel $m$ to sense.  The probability that $u_m$ secondary users choose channel $m$ to sense is
\begin{equation}\label{eq:NumUser}
  \Pr(U_m=u_m)=\left(\!\!\!
  \begin{array}{c}
    N \\
    u_m
  \end{array}
  \!\!\!\right) \left( \frac{1}{M} \right)^{u_m} \left( \frac{M-1}{M} \right)^{N-u_m}_{.}
\end{equation}
The joint distribution that there are $u_1$ secondary users sensing channel $1$, $u_2$ secondary users sensing channel $2$, $\cdots$, and $u_M$ secondary users sensing channel $M$, is
\begin{eqnarray}\label{eq:JtNumUser}
	\Pr(u_1, u_2, \cdots, u_M) = \left\{
			\begin{array}{ll}
				\prod_{m=1}^M \Pr(U_m=u_m), & \mbox{if } \; \sum_{m=1}^M u_m=N \\
				0, & \mbox{otherwise}.
			\end{array} \right.
\end{eqnarray}

We next derive the conditional probability that secondary users compete for the channel after the sensing phase stops at the end of mini-slot $K_m < \bar{K}$.  The stop time $K_m < \bar{K}$ has two implications. First, it means that secondary users stop sensing channel $m$ after mini-slot $K_m$.  Second, it indicates that the estimated availability of channel $m$, $a_{m,k}$, has already exceeded the threshold $\Theta_1$. Thus these secondary users think channel $m$ is idle and are ready to access the channel for data transmission.  Note that a secondary user also stops sensing a channel $m$ when $a_{m,k} < \Theta_0$ (when it is sure that the channel is busy). We are not interested in this case, since the secondary user will back off until the next time slot.  Thus $K_m$ is defined with regard to the event $a_{m,k} > \Theta_1$. 

There are $U_m$ users sensing channel $m$ and $U_m K_m$ observations are available after mini-slot $K_m$, which is also a random variable. 
%Let $a_{m,k}^u$ denote the conditional probability that channel $m$ is believed to be idle after mini-slot $k$, when there are $u$ users sensing that channel and the channel state is idle.  
We first derive the conditional probability for event $K_m=1$, as
\begin{eqnarray}\label{eq:CondStop1}
    && \Pr(K_m=1 \; | \; U_m=u \;, \;S_m=0) = 
      \Pr(a_{m,u} \ge  \Theta_1)   \nonumber \\
     & = & \sum_{d_m^1 \in \Psi^1_{m,u}}
			\left(\!\!\!
			\begin{array}{c}
			u \\
			d_m^1
			\end{array}
			\!\!\!\right)\left[
(\epsilon_m)^{u-d_m^1}(1-\epsilon_m)^{d_m^1} \right],
\end{eqnarray}
where $d_m^1$ is the number of observations whose sensing result is $0$
%that the sensing results are $0$ 
in the first mini-slot. 

Following similar reasoning as in (\ref{eq:CondStop1}), we can obtain the conditional probability for the event that the stop time $K_m=2$ as 
\begin{eqnarray} %\label{eq:CondStop2}
    && \Pr(K_m=2 \; | \; U_m=u \;, \;S_m=0) = \Pr \left[(\Theta_0 < a_{m,u} < \Theta_1) \cap 
        (a_{m,2u} \ge \Theta_1) \right] \nonumber \\
    & = & \hspace*{-0.1in} \sum_{D_m^2 \in
\Psi^1_{m,2u}}\sum_{d_m^1 \in
       \Psi^2_{m,u}}
\hspace*{-0.1in} \left(\!\!\!
\begin{array}{c}
u \\
d_m^1
\end{array}
\!\!\!\right)
\left(\!\!\!
\begin{array}{c}
u \\
d_m^2
\end{array}
\!\!\!\right)
\left[(\epsilon_m)^{2u-D_m^2} (1-\epsilon_m)^{D_m^2}\right], %\nonumber 
\end{eqnarray}
where $\Psi^{2}_{m,k}$ is defined in (\ref{eq:Psi2})
%, is the complement of the union of sets $\Psi^0_{m,k}$ and $\Psi^1_{m,k}$, 
and $D_m^2=d_m^1 + d_m^2$.
%
%Similarly, we 
In the general case, we can derive the conditional probability for the event that the stop time is $K_m=k$ as:
\begin{eqnarray}\label{eq:CondStopk}
& & \hspace*{-0.1in} \Pr(K_m=k \; | \; U_m=u\; , \; S_m=0) \nonumber \\
& & \hspace*{-0.25in} = \Pr \left[(\Theta_0 < a_{m,u} < \Theta_1)\cap
(\Theta_0 < a_{m,2u} < \Theta_1) \cap  %\right. \nonumber \\
%&& \hspace*{0.2in}  \left. 
  \cdots \cap (\Theta_0 < a_{m,(k-1)u} < \Theta_1)
\cap (a_{m,ku}\ge \Theta_1) \right] \nonumber \\
& & \hspace*{-0.25in} = \sum_{D_m^k \in \Psi^1_{m,ku}}
\sum_{D_m^{k-1} \in \Psi^2_{m,(k-1)u}}\cdots
\sum_{d_m^1 \in \Psi^2_{m,u}} 
%\nonumber \\
%& & \hspace*{-0.0in}
\left(\!\!\!
\begin{array}{c}
u \\
d_m^1
\end{array}
\!\!\!\right)
\left(\!\!\!
\begin{array}{c}
u \\
d_m^2
\end{array}
\!\!\!\right)
\cdots
\left(\!\!\!
\begin{array}{c}
u \\
d_m^k
\end{array}
\!\!\!\right) \left[
(\epsilon_m)^{ku-D_m^k}(1-\epsilon_m)^{D_m^k} \right], %\nonumber
\end{eqnarray}
where $k = 1, \cdots, \bar{K}$ and $D_m^k=\sum_{i=1}^k d_m^i$. We will apply these results in Section~\ref{subsec:throughput} to derive the throughput of the CR network by the {\em law of total probability}.

%\smallskip
%------------------------------------------------------
\subsubsection{Improved Sensing Policy \label{subsubsec:improved}}

Under the memoryless sensing policy, some channels may not be sensed by any of the secondary users. Such an event occurs with probability $\Pr (U_m=0) = \left( \frac{M-1}{M} \right)^N $, which is sufficiently large when $M$ is large and/or the number of secondary users is close to the number of channels. Secondary users will not be able to estimate the state of a channel that nobody senses, and will neither access it in the transmission phase. Therefore, the spectrum opportunities in that channel will be wasted when such events occur. 

Motivated by this observation, we 
%can enhance the memoryless policy 
develop an {\em improved sensing} policy that attempts to reduce the chance that a channel is not sensed by any of the secondary users. The improved sensing policy incorporates a mechanism to spread secondary users to the channels. It also exploits channel state autocorrelation by considering sensing results and channel states in the previous time slot. 

By the end of the sensing phase in a time slot $t$, the secondary users compute the channel availability $a_{m,k}$ for each channel $m$. During the following transmission phase, if a secondary user transmits on channel $m$, it can obtain more accurate channel state information: if its transmission is successful, then channel $m$ is idle in time slot $t$; otherwise, channel $m$ is busy in the time slot. Such channel information can be exchanged at the beginning of the sensing phase in the next time slot. Then, we can classify the $M$ channels into three sets according to the channel states in time slot $t$, including 
\begin{itemize}
	\item The set of channels that are detected or believed to be idle, denoted by $\mathcal{B}_0(t)$.  
	\item The set of channels that are detected or believed to be busy, denoted by $\mathcal{B}_1(t)$. %, and 
	\item The set of channels whose states are not identified due to lack of time or not sensed by any of the secondary users, denoted by $\mathcal{B}_2(t)$. 
\end{itemize}
Let $|\mathcal{B}_0(t)|$, $|\mathcal{B}_1(t)|$ and $|\mathcal{B}_2(t)|$ be the cardinalities of $\mathcal{B}_0(t)$, $\mathcal{B}_1(t)$, and $\mathcal{B}_2(t)$, respectively. 

If channel $m$ is in set $\mathcal{B}_0(t)$ and the stop time on channel $m$ is less than the maximum stop time $\bar{K}$, one user among those $u_m$ users that are sensing this channel will be randomly chosen to switch to sense another channel in the set $m \cup \mathcal{B}_1(t) \cup \mathcal{B}_2(t)$ in time slot $(t+1)$. %Note that the remaining $u_m - 1$ users in $U_m$ will remain on channel $m$ and prepare to access it for data transmission.  
If channel $m$ is in set $\mathcal{B}_1(t)$ and the stop time on channel $m$ is less than the maximum stop time $\bar{K}$, the secondary users that are sensing this channel will randomly choose a channel in $m \cup \mathcal{B}_2(t)$ to sense in time slot $(t+1)$.  With the above mechanism that reassigns secondary users to channels based on the sensing results in the previous time slot, we can reduce the chance that a licensed channel is not sensed by any of the users.  This approach achieves the {\em load balancing} effect since it attempts to spread out secondary users to the channels.

%------------------------------------------------------
\subsection{Transmission Phase \label{subsec:transph}}

We adopt the $p$-persistent CSMA protocol for data channel access for secondary users during the data transmission phase. Under this protocol, a secondary user delays its transmission when the channels are busy. 
%waits until channel becomes idle. 
Once one or more channels are detected idle, the secondary user will attempt to access the idle channel(s) for data transmission with probability $p$.  We consider the heavy load domain, where each secondary user always has data to send to every other secondary user. The following two cases are investigated for opportunistic spectrum access for secondary users. 

%------------------------------------------------------
\subsubsection{Case 1}

Once the estimate of channel $m$, i.e., $a_{m,k}$, exceeds threshold $\Theta_1$, each of the secondary users sensing channel $m$ will send an RTS packet on channel $m$ with probability $p$, to contend for the transmission opportunity on this channel. If there is only one secondary user that sends RTS, then it wins the channel; if there is no secondary user that sends RTS, then the channel will not be accessed and will be wasted; if there are more than one RTS packets sent on channel $m$, there is collision and none of the secondary users can use the channel. 

%In the $p$-persistent CSMA, if we 
We define $P_m^{idle}$, $P_m^{succ}$ and $P_m^{coll}$ as the probability that there is no RTS transmission on channel $m$, the probability that exactly one secondary user successfully transmits an RTS on channel $m$, and the probability that there is collision on channel $m$ when multiple RTS packets are transmitted, respectively.  Recall that $U_m$ is the number of secondary users that choose channel $m$ to sense. This set of secondary users also attempt to access channel $m$ if it is found idle. With $p$-persistent CSMA, it follows that 
%for all $m$, 
\begin{eqnarray}\label{eq:ProbAccess}
%\left\{
%\begin{array}{l}
&& \hspace*{-0.4in} P_m^{idle}(U_m) = (1-p)^{U_m} \label{eq:pmidle} \\
&& \hspace*{-0.4in} P_m^{succ}(U_m) = U_m\times p\times(1-p)^{U_m-1} \label{eq:pmsucc} \\
&& \hspace*{-0.4in} P_m^{coll}(U_m) = 1 - P_m^{idle}(U_m) - P_m^{succ}(U_m) 
%\nonumber \\
%&& \hspace*{0.23in} 
                    = 1-(1-p)^{U_m}-U_m\times p\times(1-p)^{U_m-1}_{.} \label{eq:pmcoll} 
%\end{array}\right. \hspace*{-0.1in}
\end{eqnarray}

%------------------------------------------------------
\subsubsection{Case 2}

We assume that the CR users can transmit data over more than one channels using the channel bonding/aggregation techniques~\cite{Su08JNL, Corderio06}. In this case, every secondary user keeps on sensing the channel until the channel state is identified or until the end of the sensing phase. At the beginning of the transmission phase, the set of idle channels are identified and are know to all the secondary users.  Then every secondary user will transmit an RTS packet with probability $p$ on the control channel, to contend for the entire set of idle channels.  If there is only one secondary user that sends RTS on the control channel, it wins the entire set of idle channels. Otherwise, the idle channels will be wasted (i.e., when no RTS is sent, or more than one RTS are sent on the control channel).  

We define $P^{idle}$, $P^{succ}$ and $P^{coll}$ as the probability of no RTS transmission on the control channel, the probability that exactly one RTS sent on the control channel, and the probability of collision on the control channel, respectively.
%Once one of the estimation of channels exceeds the threshold $\Theta_1$, all of the $N$ secondary users send RTS packets on the control channel, each with probability $p$, to contend for the transmission opportunity on this channel.  
For $p$-Persistent CSMA, we have 
\begin{eqnarray}\label{eq:ProbAccess2}
%\left\{
%\begin{array}{l}
&& \hspace*{-0.4in} P^{idle}(N) = (1-p)^{N} \label{eq:pmidle2} \\
&& \hspace*{-0.4in} P^{succ}(N) = N\times(1-p)^{N-1} \label{eq:pmsucc2} \\
&& \hspace*{-0.4in} P^{coll}(N) = 1 - P^{idle}(N) - P^{succ}(N) 
%\nonumber \\
%&& \hspace*{0.23in} 
                    = 1-(1-p)^{N}-N\times p \times (1-p)^{N-1}_{.} \label{eq:pmcoll2} 
%\end{array}\right. \hspace*{-0.1in}
\end{eqnarray}

%------------------------------------------------------
\subsection{Interference Analysis \label{subsec:interference}}

One of the main challenges in designing a CR network MAC protocol is how to balance the tension between maximizing the capacity of secondary users and protecting primary users from harmful collisions. Let $\gamma_m \in [0,1]$ be the maximum tolerable collision probability to primary users on channel $m$: $\gamma_m=0$ means that no secondary transmission is allowed, while $\gamma_m=1$ means that secondary users have the same privilege as primary users when accessing the channels.
The probability of collision caused by secondary users to primary users should be kept below $\gamma_m$. 

We first derive the conditional probability that channel $m$ is miss detected to be idle by $u$ secondary users after mini-slot $k$, as follows.
\begin{eqnarray} \label{eq:CondIntf}
& & \hspace*{-0.2in} \Pr(K_m=k \; | \; U_m=u,S_m=1) \nonumber \\
& & \hspace*{-0.4in} = \sum_{D_m^k \in \Psi^1_{m,ku}}
\sum_{D_m^{k-1} \in \Psi^2_{m,(k-1)u}}\cdots
\sum_{d_m^1 \in \Psi^2_{m,u}}
%\nonumber \\ 
%&& \hspace*{-0.2in}
\left(\!\!\!
\begin{array}{c}
u \\
d_m^1
\end{array}
\!\!\!\right)
\left(\!\!\!
\begin{array}{c}
u \\
d_m^2
\end{array}
\!\!\!\right)
\cdots
\left(\!\!\!
\begin{array}{c}
u \\
d_m^k
\end{array}
\!\!\!\right)(\delta_m)^{D_m^k}(1-\delta_m)^{ku-D_m^k}_{.}
\end{eqnarray}
In Case 1, the idle channels are accessed by different secondary users. The probability that secondary users collide with primary users on channel $m$ is
\begin{eqnarray}\label{eq:ProbIntf}
P_{m,1}^{intf}
%& = & 
=  \sum_{k=1}^{\bar{K}}\sum_{u=0}^N \Pr(K_m=k \; | \; U_m=u,S_m=1) \times 
           %\nonumber \\
           %&   & 
           \Pr(U_m=u) \times \left[ P_m^{succ}(u)+P_m^{coll}(u) \right].  
\end{eqnarray}
In Case 2, a winning secondary user takes all the idle channels using the channel bonding/aggregation technique. The probability that secondary users collide with primary users on channel $m$ is
\begin{eqnarray}\label{eq:ProbIntf2}
P_{m,2}^{intf} =
%& = & 
  \sum_{k=1}^{\bar{K}}\sum_{u=0}^N \Pr(K_m=k \; | \; U_m=u,S_m=1) \times 
%           \nonumber \\
%           &   & 
           \Pr(U_m=u) \times P^{succ}(N).  
\end{eqnarray}

For primary user protection, the probability of secndary users causing collision with primary users on channel $m$ should be kept lower than or equal to 
%the maximum allowable probability 
$\gamma_m$, i.e.,
\begin{eqnarray}
P_m^{intf} \le \gamma_m.
\end{eqnarray}
This constraint is used to set the channel access probability $p$ for the $p$-persistent CSMA protocol.

%------------------------------------------------------
\subsection{Throughput Analysis \label{subsec:throughput}}

Based on previous analysis, the expected throughput of the proposed CR MAC protocols adopting the two sensing policies, can be derived after the system attains steady state.  Without loss of generality, we ignore the time spent on RTS/CTS exchanges, which can be approximated by a fixed amount of overhead. 

In Case 1, the expected throughput of channel $m$ that is sensed by $u$ users, denoted by $\Lambda^1_m(u)$, can be derived as
%\begin{eqnarray}
%&& \hspace*{-0.4in} \Lambda_m(u) = \sum_{k=1}^{\bar{K}}\Pr(K_m=k\; | \; U_m=u,S_m=0) \times \nonumber \\
%&& \hspace*{0.35in} P_m^{succ}(u) \times R_m \times \frac{(\bar{K}-k) T_{ms} + T_{data}}{T_s},
%\end{eqnarray}
\begin{eqnarray}
\Lambda^1_m(u) 
%&=& 
= \sum_{k=1}^{\bar{K}}\Pr(K_m=k\; | \; U_m=u,S_m=0) \times %\nonumber \\
    %&& 
    R_m \times \frac{1}{T_s} \times \left[ (\bar{K}-k) T_{ms} + T_{data} \right],
\end{eqnarray}
where $R_m$ is the data rate of channel $m$, and $T_s$ is the time slot duration given in (\ref{eq:ts}).  
%The aggregate throughput for the CR network, denoted by $\Lambda$, is
%\begin{eqnarray}
%&& \hspace*{-0.2in} \Lambda = \sum_{u_1,\cdots,u_m}\Pr(u_1,\cdots,u_m)\sum_{m=1}^M \Lambda_m(u_m) \Pr(S_m=0) \nonumber \\
%&& \hspace*{-0.06in} = \sum_{u_1,\cdots,u_m}\Pr(u_1,\cdots,u_m)\sum_{m=1}^M \Lambda_m(u_m)\times \zeta_m.
%\end{eqnarray}

%Denote $\vec{S}=[S_1,S_2,\cdots,S_M]$ and $\vec{U}=[U_1,U_2,\cdots,U_M]$ as the network state and the user allocation state, respectively. 
Let $\vec{U}=[U_1, U_2, \cdots, U_M]$ denote the secondary user {\em sensing state vector}, where each element $U_m$ represents the number of secondary users that choose channel $m$ to sense and access.  The aggregate throughput for the CR network, denoted by $\Omega_1$, is
\begin{eqnarray}
%\Lambda = \sum_{\vec{U}\in R^M}\Pr(\vec{U})\sum_{\vec{S}\in 2^M}\Pr(\vec{S})\sum_{m=1}^M I_{[S_m=0]}\Lambda_m(U_m),
\Omega_1 = \sum_{\vec{U}} \Pr(\vec{U}) \sum_{\vec{S}} \Pr(\vec{S}) \sum_{m=1}^M \left( I_{[S_m=0]}\Lambda^1_m(u) \times P_m^{succ}(u) \right),
\end{eqnarray}
where $\vec{S}$ is the channel state vector defined in Section~\ref{sec:model}, $P_m^{succ}(u)$ is given in (\ref{eq:pmsucc}) and $I_{[S_m=0]}$ is an indicator that channel $m$ is idle, i.e.,
%(i.e., $I_{[S_m=0]}=1$ when $S_m=0$; otherwise it is $0$).
\begin{equation}
  I_{[S_m=0]} = \left\{ \begin{array}{ll} 
                   1, & \mbox{if } S_m=0 \\
                   0, & \mbox{otherwise}.
                        \end{array} \right.
\end{equation} 

In Case 1, the sensing process on channel $m$ can stop early if the estimate of channel availability $a_{m,k}$ exceeds threshold $\Theta_1$ or drops below the threshold $\Theta_0$. In the former case, the remaining mini-slots can be used to transmit data. In Case 2, all CR users wait till the beginning of the transmission phase, and then contend for the idle channels by sending RTS packets on the control channel. 
%The user winning the transmission opportunities may still need to sense the original channel. 
%For simplicity, we assume the 
The winning secondary user's data transmissions start at the beginning of the transmission phase (i.e., after $\bar{K}$ mini-slots). 
%Similarly, we can get $\Lambda^2_m(u)$
We can derive the 
%CR network 
throughput for channel $m$ as follows. 
\begin{eqnarray}
\Lambda^2_m(u) = 
%&=& 
\sum_{k=1}^{\bar{K}}\Pr(K_m=k\; | \; U_m=u,S_m=0) \times 
%\nonumber \\
%    &&
R_m \times \frac{T_{data}}{T_s},
\end{eqnarray}
The aggregate throughput for the CR network, denoted by $\Omega_2$, is
\begin{eqnarray}
%\Lambda = \sum_{\vec{U}\in R^M}\Pr(\vec{U})\sum_{\vec{S}\in 2^M}\Pr(\vec{S})\sum_{m=1}^M I_{[S_m=0]}\Lambda_m(U_m),
\Omega_2=\sum_{\vec{U}} \Pr(\vec{U}) \sum_{\vec{S}} \Pr(\vec{S}) \sum_{m=1}^M \left( I_{[S_m=0]}\Lambda^2_m(u) P^{succ}(N) \right).
\end{eqnarray}

%------------------------------------------------------
\section{Simulation Results \label{sec:sim}}
%------------------------------------------------------

%------------------------------------------------------
\subsection{Simulation Settings \label{subsec:simset}}

We evaluate the performance of the proposed CR MAC protocol using a customized simulator developed with MATLAB.  We compare the following four schemes in the simulations: 
\begin{itemize}
	\item A simple random sensing scheme that each user chooses one channel to sense with equal probability, termed {\em Random} in the plots. 
	\item The negotiate sensing scheme presented in~\cite{Su08JNL}, termed {\em Negotiate} in the plots. 
	\item The memoryless sensing scheme as described in Section~\ref{subsubsec:memoryless}. In the figures, {\em Memoryless}1 refers to transmission scheme Case 1 (i.e., idle channels are accessed by different secondary users, see Section~\ref{subsec:transph}), and {\em Memoryless}2 refers to transmission scheme Case 2 (i.e., idle channels are accessed by a winning secondary user using channel bonding/aggregation techniques~\cite{Corderio06}).  
	\item The improved sensing scheme presented in Section~\ref{subsubsec:improved}. In the figures, {\em Improved}1 refers to transmission scheme Case 1, and {\em Improved}2 refers to transmission scheme Case 2.
\end{itemize}
We choose the negotiate sensing scheme since it adopts a similar network model and assumptions.  With this scheme, different secondary users attempt to select distinct channels to sense by overhearing the control packets on the control channel~\cite{Su08JNL}. One of the major differences between negotiate sensing and the proposed schemes in this paper, is that negotiate sensing does not consider spectrum sensing errors in the MAC protocol design. 

The simulation parameters are summarized in Table~\ref{tb:Parameter}, which follow the typical values used in~\cite{Su08JNL}. We run each simulation scenario for 10 times with different random seeds. Each point in the plots shown in this section is the average of $10$ simulation runs. We plot $95\%$ confidence intervals as error bars on the simulation curves, which are negligible in all the figures.  

%------------------------------------------------
\begin{table}
\begin{center}
\caption{Simulation Parameters}
\begin{tabular}{l l l}
\hline 
{\em Symbol} & {\em Value} & {\em Definition} \\
\hline 
$T_{ms}$ & 9 $\mu$s & mini-slot interval \\
$T_s$ & 1.89 ms & time slot interval \\
$M$ & 5 & number of licensed channels \\
$N$ & 8 & number of secondary users \\
$\eta$ & 0.3 & utilization of the licensed channels \\
$\epsilon$ & 0.3 & probability of false alarm \\
$\delta$ & 0.3 & probability of miss detection \\
$R$ & 1 Mb/s & data rate of each licensed channel \\
$\Theta_1$ & 0.8 & upper threshold for channel decision \\
$\Theta_0$ & 0.2 & lower threshold for channel decision \\
$\bar{K}$ & 5 & maximum stop time for channel sensing \\
\hline 
\end{tabular}
\label{tb:Parameter}
\end{center}
\end{table}
%-------------------------------------------------

%------------------------------------------------------
\subsection{Simulation Results \label{subsec:simre}}

We first verify our throughput analysis presented in Section~\ref{sec:crmac}. In 
%Figs.~\ref{fig:UtilityRes},
Figs~\ref{fig:FalseAlarmRes} and~\ref{fig:MissDetecRes}, we plot the throughputs for the CR MACs incorporating the memoryless sensing policy and the improved sensing policy, with both simulation and analysis curves (dashed curves). We observe that the simulation and analysis curves for the memoryless sensing CR MACs overlap completely with each other, indicating that our analysis is exact.
%highly accurate.  
Furthermore, although there is a gap between the simulation and analysis curves for the CR MACs with the improved sensing policy, the gap is generally very small.  The gap is actually due to an approximation we used for the secondary user sensing state vector $\vec{U}$, for which deriving the exact form is non-trivial. In the analysis, we assume that the probability is 0 that a channel is not sensed by any secondary user.  We find the analysis can serve as a tight upper bound for the CR MAC throughput performance when the improved sensing policy is incorporated.  

\begin{figure}
\centering  
\includegraphics[width=3.3in, height=2.1in]{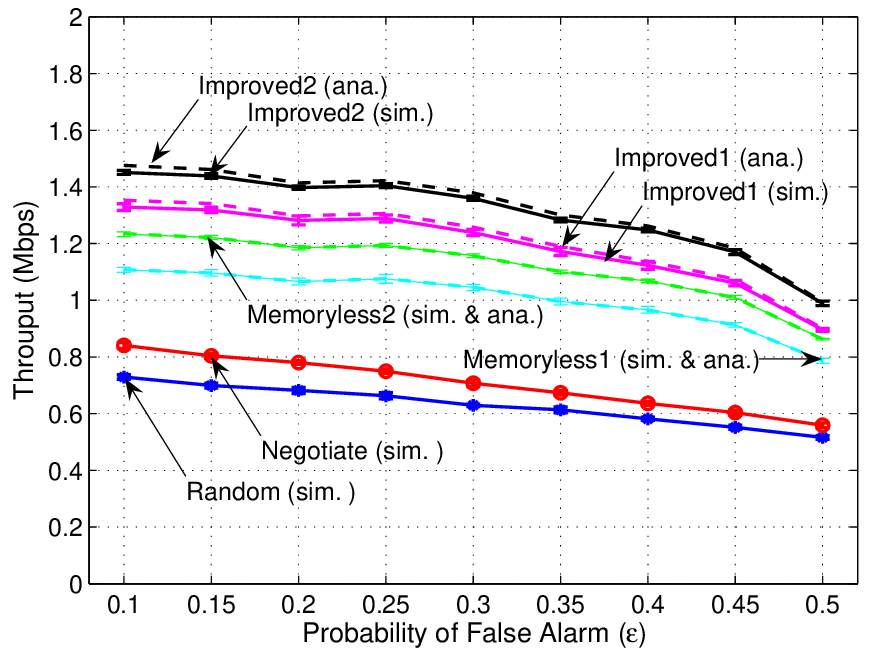}
\caption{Throughput versus false alarm probability (with $95\%$ confidence intervals for the simulation results).}
\label{fig:FalseAlarmRes}
\end{figure}

\begin{figure}
\centering  
\includegraphics[width=3.3in, height=2.1in]{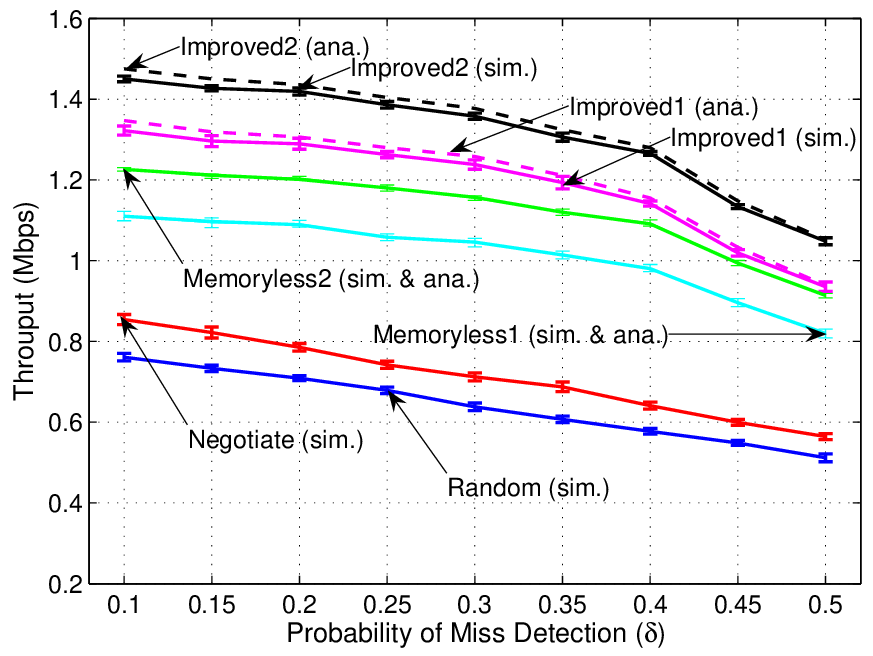}
\caption{Throughput versus miss detection probability (with $95\%$ confidence intervals for the simulation results).}
\label{fig:MissDetecRes}
\end{figure}

We next investigate the impact of sensing errors on the CR MAC performance. We assume identical false alarm probabilities $\epsilon_m = \epsilon$, and identical miss detection probabilities $\delta_m = \delta$ for all the licensed channels. In Fig.~\ref{fig:FalseAlarmRes}, we plot the throughputs obtained by the four schemes versus the false alarm probability $\epsilon$.  Specifically, we fix $\delta$ at 0.3 and increase $\epsilon$ from $0.1$ to $0.5$.  Intuitively, a higher false alarm probability results in lower probability for secondary users to exploit the transmission opportunities in the licensed channels. This is illustrated in the figure, as all the four throughput curves decrease as $\epsilon$ is increased. The improved sensing MAC achieves the best performance, with about 10\% gain over the memoryless sensing MAC and about 200\% gain over the two existing approaches. The advantage of channel bonding/aggregation is also demonstrated in the figure, where Case 2 transmission scheme always achieves higher throughput than Case 1 scheme. 

In Fig.~\ref{fig:MissDetecRes}, we examine the impact of miss detection probability  $\delta$ on the CR network throughput. In these simulations, we fix $\epsilon$ at 0.3 and increase $\delta$ from $0.1$ to $0.5$.  We find that the miss detection error has small impact on the throughputs of the random sensing and negotiate sensing protocols, since miss detection errors are not considered in the design of these protocols.  However, both our proposed CR MAC schemes achieve considerable throughput gains over the random sensing and negotiate sensing schemes.

In Fig.~\ref{fig:UtilityRes}, we plot the throughput of the four schemes under different channel utilization values ranging from $0.3$ to $0.7$.  As utilization of the licensed channels is increased, the transmission opportunities for secondary users are clearly reduced. Therefore the four curves are all decreasing function of $\eta$.  The improved policy with transmission scheme Case 2 achieves the best performance among the four schemes, while random sensing has the poorest performance.  When the channel utilization is $\eta=0.3$, the improved policy achieves a $10\%$ gain in throughput over the memoryless sensing policy. We also plot the upper bound on the CR network throughput, as given by the channel idle probability in (\ref{eq:ProbIdle}). When the channel utilization is low, the improved policy with transmission scheme Case 2 can achieve a throughput very close to the upper bound. The gap between the upper bound and the achievable throughput increase when the primary users get more busy. 

\begin{figure}
\centering  
\includegraphics[width=3.3in, height=2.1in]{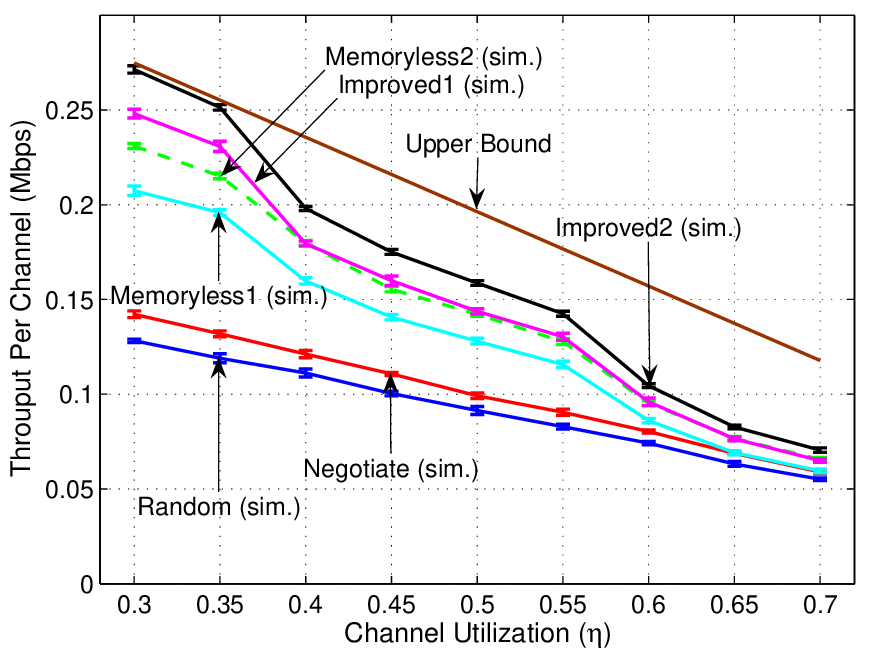}
\caption{Throughput versus channel utilization (with $95\%$ confidence intervals for the simulation results).}
\label{fig:UtilityRes}
\end{figure}

In Fig.~\ref{fig:collision}, we plot the collision probability caused by secondary transmissions to primary users, when the maximum allowable collision probability is set as $\gamma=3.5\%$. We plot the measured collision probabilities in the simulations when the channel utilization is increased from 30\% to 70\%. It can be seen that the collision probabilities of random and negotiate sensing schemes increases along with $\eta$ and soon exceed the 3.5\% threshold. On the other hand, the collision probabilities of the proposed schemes are kept around 2.5\% for the entire range of $\eta$ examined. 
%Therefore, we can see Negotiate sensing and Random Sensing will exceed the preset maximum allowable collision probability when channel utilization is larger than 0.45.

\begin{figure}
\centering
\includegraphics[width=3.3in, height=2.1in]{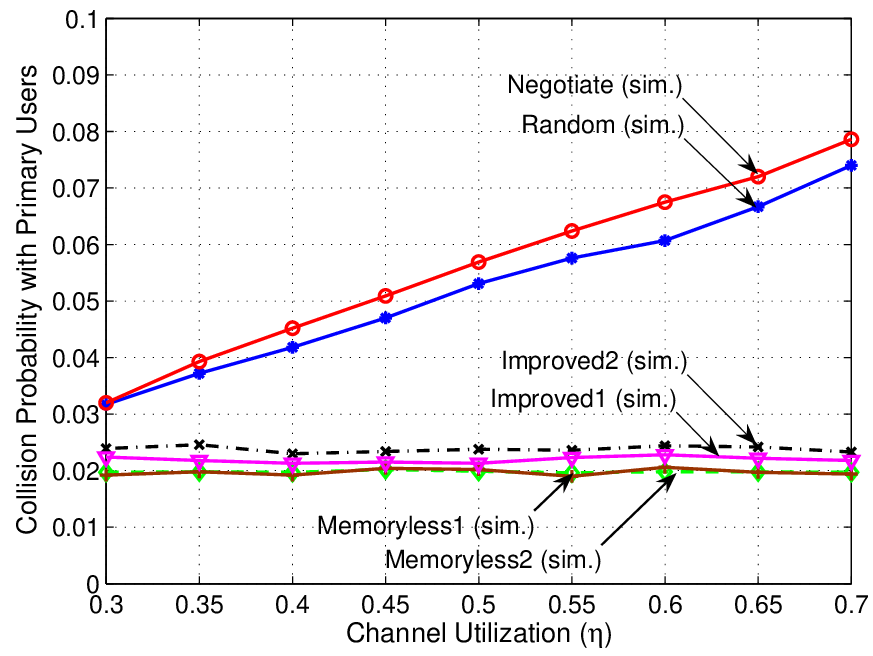}
\caption{Collision probability with primary users when the maximum tolerable collision probability is $\gamma = 3.5\%$.}
\label{fig:collision}
\end{figure}

Finally, we plot the throughput of the primary users in Fig.~\ref{fig:primarythput}. The primary user throughput curves for all the four schemes increase when the channel utilization $\eta$ is increased. The gap between the curves of the proposed schemes and those of random and negotiate sensing schemes, is due to the different collision rates secondary users introduce to primary users under these schemes (see Fig.~\ref{fig:collision}). As $\eta$ is increased, the proposed schemes introduces relatively constant collision rates to primary users (i.e., around 2.5\%), while the random and negotiate sensing schemes introduce increasingly higher collision rates to primary users, which degrade the throughput of primary users.  

%Improved sensing and Memoryless sensing impose less interference than Negotiate sensing and Random sensing when channel utilization is larger. Compare Improved sensing with Memoryless sensing, the inference with primary user is nearly the same, but improved sensing scheme has better performance than Memoryless sensing in secondary network.

\begin{figure}
\centering  
\includegraphics[width=3.3in, height=2.1in]{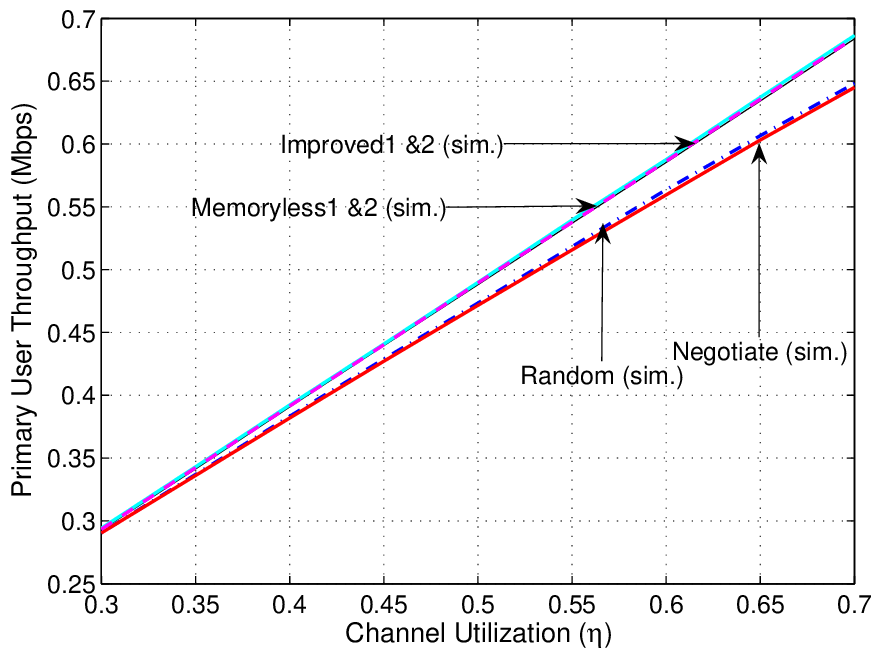}
\caption{Total throughput of primary users when they become more active.}
\label{fig:primarythput}
\end{figure}

%------------------------------------------------------
\section{Related Work \label{sec:related}}
%------------------------------------------------------

%The high potential of CR has stimulated a flurry of exciting activities in engineering, economics, and regulatory communities recently in searching for more efficient spectrum management policies and sharing techniques.  The mainstream CR research has considered CR as ``spectrum agile radio'' that enables dynamic spectrum access to exploit transmission opportunities in licensed spectrum bands~\cite{Zhao07c, Zhao09}. 

CR has been considered as a ``spectrum agile radio'' that enables dynamic spectrum access to exploit transmission opportunities in licensed spectrum bands~\cite{Zhao07c, Zhao09}.  Several CR MAC protocols have been proposed in the literature~\cite{Su08JNL, Su08WM, Su07CISS, Le08, Zhao07JNL, Jia08JNL, Hamdaoui08, Hsu07, YChen08}. In~\cite{Le08}, Le and Hossain propose a MAC protocol for opportunistic spectrum access in CR networks. Two channel selection schemes are proposed: uniform channel selection and spectrum opportunity-based channel selection. The latter considers the probability of spectrum availability and selects each channel with different probabilities based on the estimation of spectrum availability.  

A decentralized cognitive MAC protocol is developed in~\cite{Zhao07JNL} that allows secondary users to explore spectrum opportunities without a central coordinator or a dedicated control channel. However, the implementation is complicated and hardware demanding. This is because each secondary user needs to be equipped with multiple sensors to detect the availability of each licensed channel.  
%A. Motamedi \cite{Motamedi07} model the occupancy of channel by PU (in terms of the number of time slots) as geometric random variables which is limited in reality.

In a recent work~\cite{Su08JNL}, Su and Zhang propose a negotiation-based sensing policy (NSP), in which a secondary user knows which channels are already sensed and will choose a different channel to sense. In~\cite{Jia08JNL}, the authors consider two types of hardware constraints: sensing constraint and transmission constraint. In~\cite{Hamdaoui08}, based on the information obtained by a delegate secondary user, each secondary user group selects and switches to the best data channel for data communication during the next period. In~\cite{Hsu07}, the authors describe a policy such that a secondary user selects the channel that has the highest successful transmission probability to access.  
%L. Ma consider three separate kinds of transceivers. 
Many prior works~\cite{Le08, Su08JNL, Jia08JNL, Motamedi07} assume perfect channel sensing, within which secondary users can always sense the channel correctly. Sensing errors are not considered. 
%Although sensing errors are considered in a recent work~\cite{Zhao07JNL} in the simulation studies, no analysis is presented for the impact of sensing error on the CR network performance in this paper.

The joint design of opportunistic spectrum access and sensing policies is studied in a recent work~\cite{YChen08} in the presence of sensing errors. The authors develop a separation principle that decouples the designs of sensing and access policy. This interesting study is based on a constrained partially observable Markov decision process (POMDP) formulation and thus has an exponentially growing computational complexity~\cite{YChen08}.

%{\bf More references and detailed discussion of related work.}

% of obtaining the optimal sensing policy is O(NT), which grows exponentially with the horizon length T.
%Separation Principle for Opportunistic Spectrum Access in the Presence of Sensing Errors}
%For the OSA network example, a separation principle based on a POMDP framework has been established in [23] and [24] that leads to a closed-form characterization of the optimal access strategy jointly designed with the sensing strategy for any operating point ƒÂ ¸ (0, 1) of the spectrum detector.

%------------------------------------------------------
\section{Conclusion \label{sec:con}}  
%------------------------------------------------------

We 
%presented a sensing error aware CR MAC protocol and 
studied the problem of design and analysis of MAC protocol for CR networks in this paper. In particular, we proposed and analyzed two opportunistic multi-channel MAC protocols, adopting a memoryless sensing policy and an improved sensing policy, respectively.  The impact of imperfect sensing (in the forms of miss detection and false alarm) are explicitly considered in the CR MAC design.  We developed analytical models to evaluate the performance of the proposed protocols.  Our simulation study demonstrates the accuracy of the analysis, as well as the superior throughput performance of the proposed CR MACs over existing approaches.

%------------------------------------------------------
\section*{Acknowledgment}
%------------------------------------------------------

This work is supported in part by the US National Science Foundation (NSF) under Grants CNS-0953513, ECCS-0802113, and IIP-1032002, % and CNS-0855251, 
and through the NSF Wireless Internet Center for Advanced Technology (WICAT) at Auburn University. Any opinions, findings, and conclusions or recommendations expressed in this material are those of the author(s) and do not necessarily reflect the views of the Foundation.

\end{document}